\documentclass[11pt]{article}

\usepackage{fullpage}
\usepackage{times}
\usepackage{amsmath,amsfonts,amssymb}
\usepackage{amsthm}
\usepackage{pdfpages}

\newtheorem{theorem}{Theorem}[section]
\newtheorem{corollary}[theorem]{Corollary}
\newtheorem{lemma}[theorem]{Lemma}
\newtheorem{claim}[theorem]{Claim}

\newtheorem{proposition}[theorem]{Proposition}

{\theoremstyle{remark} }
{\theoremstyle{definition} \newtheorem{definition}[theorem]{Definition}}

\newenvironment{proofof}[1]{\begin{proof}[Proof of #1]}{\end{proof}}
\newenvironment{proofsketch}{\begin{proof}[Proof Sketch]}{\end{proof}}

\newcommand{\comment}[1]{}
\newcommand{\remove}[1]{}
\newcommand{\suppress}[1]{}

\newcommand{\NN}{{\mathbb{N}}}

\newcommand{\LP}{{\hbox{\sc lp}}}
\newcommand{\IP}{{\hbox{\sc ip}}}

\newcommand{\last}{{\hbox{last}}}

\DeclareMathOperator{\argmax}{argmax}
\DeclareMathOperator{\argmin}{argmin}

\newcommand{\eps}{\epsilon}

\begin{document}

\title{A Constant Factor Approximation Algorithm for
Reordering Buffer Management}

\author{Noa Avigdor-Elgrabli\thanks{Computer Science Department,
Technion---Israel Institute of Technology, Haifa 32000, Israel.
Email: {\tt noaelg@cs.technion.ac.il}.}
\and
Yuval Rabani\thanks{The Rachel and Selim Benin School of
Computer Science and Engineering and the Center of Excellence
on Algorithms, The Hebrew University of Jerusalem, Jerusalem
91904, Israel. Email: {\tt yrabani@cs.huji.ac.il}.
Research supported by ISF grants 1109-07 and 856-11 and
by BSF grant 2008059.}
}

\date{}

\setcounter{footnote}{3}
\maketitle

\begin{abstract}
In the reordering buffer management problem (RBM) a sequence
of $n$ colored items enters a buffer with limited capacity $k$.
When the buffer is full, one item is removed to the output sequence,
making room for the next input item. This step is repeated until
the input sequence is exhausted and the buffer is empty. The
objective is to find a sequence of removals that minimizes the
total number of color changes in the output sequence. The
problem formalizes numerous applications in computer and
production systems, and is known to be NP-hard.

We give the first constant factor approximation guarantee for RBM.
Our algorithm is based on an intricate ``rounding" of the solution
to an LP relaxation for RBM, so it also establishes a constant upper
bound on the integrality gap of this relaxation. Our results improve
upon the best previous bound of $O(\sqrt{\log k})$ of Adamaszek
et al. (STOC 2011) that used different methods and gave an online
algorithm. Our constant factor approximation beats the super-constant
lower bounds on the competitive ratio given by Adamaszek et al. This
is the first demonstration of an offline algorithm for RBM that is provably
better than any online algorithm.
\end{abstract}

\thispagestyle{empty}
\newpage
\setcounter{page}{1}

\section{Introduction}

\paragraph{Problem statement and motivation.}
In the reordering buffer management problem (RBM) a sequence
of $n$ items of colors $c(1),c(2),\dots,c(n)$ (taken from a finite
set of colors $C$) enters a buffer with capacity $k\in\NN$. When
the buffer is full, one item is removed, making room for the next
input item. This step is repeated until the input sequence is
exhausted and the buffer is empty. Thus, the buffer can be used
to permute the input sequence in a limited way. In the permuted
output sequence, we are interested in the number of times there
is a color change between adjacent positions. Out of all feasible
solutions, the objective is to find a sequence of removals that
batches items of the same color and minimizes the total number
of color changes in the output sequence.

Introduced in~\cite{RSW02}, this elegant model formalizes a wide
scope of resource management problems in production engineering,
logistics, computer systems, network optimization, and information
retrieval (see, e.g.,~\cite{RSW02,BB02,KRSW04,GSV04}).
For example, one of the motivating examples of~\cite{RSW02} is
batching cars by color in the paint shop of a car manufacturing plant
to minimize the consumption of paint solvent used to wash spray guns
each time the paint color is changed between two consecutive cars.
Naturally, the buffer capacity is limited by physical space constraints,
and the incoming stream of cars is dictated by the schedule of the
assembly line. (This particular application is part of the ROADEF
Challenge 2005 of the French Operations Research Society~\cite{ROADEF2005},
see also~\cite{GSV04}.)
Generally, in computer systems and production systems buffers
are often prepended to subsystems to facilitate better control of
their input (see~\cite{SGG09,LM09}), so understanding how to optimize
buffer utilization is a fundamental and important problem.

\paragraph{Our results.}
We give the first constant factor approximation guarantee for RBM,
improving on the best previous bound of $O(\sqrt{\log k})$~\cite{ACER11}.
Our algorithm is based on ``rounding" the solution to a linear programming
relaxation that we recently proposed~\cite{AR10}. Thus, our work also
establishes an $O(1)$ upper bound on the integrality gap of this
relaxation, improving upon the best previous bound of $O(\sqrt{\log k})$
(which is not explicit in~\cite{ACER11}, but can be derived from their work).
Most previous work on RBM (including the above-mentioned~\cite{AR10,ACER11})
discussed online algorithms. There are recent lower bounds on the competitive
ratio of $\Omega(\sqrt{\log k/\log\log k})$ for deterministic algorithms, and
$\Omega(\log\log k)$ for randomized algorithms
(against the oblivious adversary)~\cite{ACER11}. Thus, our (deterministic) algorithm
shows, for the first time, that an efficient offline RBM algorithm can beat any online
algorithm. We note that in some applications, e.g. the paint shop sequencing
problem mentioned above, the setting enables an offline computing of a
good solution. Moreover, proving strong upper bounds on the integrality
gap of a natural linear programming relaxation seems to be one of the major
stumbling blocks on the way to design randomized online algorithms that
beat the deterministic lower bound.

Our algorithm for transforming a fractional solution into an integer one, without
increasing the cost by more than a constant factor, works in phases. A phase
starts at the time reached by the previous phase. A phase has a time horizon
target, which is the time that the fractional solution increases its cost by some
small constant factor. The goal of a phase is to reach the target by evicting a
constant number of colors. This goal might be impossible to achieve. In such
a case, we use an intricate charging scheme that chooses colors to evict and
charges their eviction to the past fractional cost of other colors. The main
difficulty in the analysis is to prove that the conditions under which the simple
strategy fails to reach the target imply that the charging scheme can be used
to bridge the gap. Our proofs involve illuminating observations on the structure
of fractional RBM solutions.

\paragraph{Previous work.}
As mentioned above, RBM was introduced in~\cite{RSW02}, who gave an
$O(\log^2 k)$-competitive online algorithm for the problem. The guarantee
was improved through a sequence of papers~\cite{EW05,AR10,ACER11},
culminating in the $O(\sqrt{\log k})$ bound of~\cite{ACER11}. This was
also the best known approximation guarantee for RBM prior to our work.
While the algorithms evolved gradually, each paper uses completely different
tools of analysis. The $O(\log k / \log\log k)$-competitive analysis of~\cite{AR10}
applies a dual fitting argument, using the same relaxation that we use in this paper.
The later and better result of~\cite{ACER11} does not use linear programming.
However, their proof can be modified to show that the $O(\sqrt{\log k})$ bound
on the competitive ratio also holds when competing against a fractional
adversary, and therefore the integrality gap of the~\cite{AR10} relaxation
is $O(\sqrt{\log k})$. As mentioned above, in~\cite{ACER11} lower bounds
on the competitive ratio of $\Omega(\sqrt{\log k/\log\log k})$ and
$\Omega(\log\log k)$, respectively, were established for deterministic and
randomized online algorithms, respectively.

Beyond the implications of the online setting, not much was known about the
offline case prior to our work. Recent work shows that the problem is
NP-hard~\cite{CMSS10,AKM10}. Allowing resource augmentation,~\cite{CMSS10}
give, for every $\eps > 0$, an $O(1/\eps)$-approximation algorithm for RBM
with the caveat that the approximation algorithm is allowed to use a buffer of
size $(2+\eps)\cdot k$. This strengthens a similar result implicit in~\cite{EW05},
proving this for their online algorithm, but only for $\eps = 2$. The
paper~\cite{EW05} also
shows that the optimum for a buffer of size $k$ can be at most a factor of
$O(\log k)$ larger than the optimum for a buffer of size $4k$. In~\cite{Abo08},
a matching  lower bound of $\Omega(\log k)$ was established, so the above
resource augmentation arguments cannot yield constant factor approximation
guarantees for RBM.

There are simple constant factor approximation algorithms for the complement
objective of maximizing the number of adjacent pairs with no color change in
the output sequence~\cite{KP04,BL07} (the constants are $20$ and $9$, respectively).
The minimization version that we consider
here seems more adequate for the applications in mind, and it also seems more
challenging. Clearly, if we expect successful batching into relatively long
monochromatic sequences, then guarantees on the complement maximization
objective do not guarantee good performance in terms of the minimization
objective.

For some applications, it is suitable to use more general cost functions to
measure the cost of color changes in the output sequence. In particular,
non-uniform costs, where the cost of switching to a color depends on
the color, were discussed in~\cite{EW05,AR10,ACER11}. Metric costs,
where the cost of switching between colors is determined by a metric on
the colors, were discussed for the line metric in~\cite{KP06,GS07} and for
general metrics in~\cite{ERW07}. None of these models are known to have
constant factor polynomial time approximation algorithms.

\section{The Algorithm}

Consider an instance ${\cal I}$ of RBM that is given by the buffer
size $k$ and by a sequence of $n$ items of colors $c(1),c(2),\dots,c(n)$.
Let $C$ denote the set of colors that appear in the sequence.
Our algorithm solves a linear programming relaxation for ${\cal I}$,
and then uses the fractional optimal solution to
derive an integer solution whose cost is at most a constant factor
greater than the fractional solution we started with. We use a time
indexed relaxation that was first defined in our previous work~\cite{AR10},
where it was used in a dual fitting analysis of an online algorithm for the
problem. For completeness, we define the relaxation and motivate it here.

We use $0$-$1$ variables $x_{i,j}$, for $i=1,2,\dots,n$ and
$j=\max\{k+1,i\},\dots,k+n$.
An assignment $x_{i,j} = 1$ indicates that the $i$-th input item is
removed from the buffer at time $j$. The reordering buffer management
problem can be expressed as an integer linear program $\IP$ on these
variables. We require the following notation. For every input item $i$,
let $\last(i)$ denote the last input item of color $c(i)$, and for $i\ne\last(i)$
let $n(i)$ denote the next input item of color $c(i)$. For notational
convenience, we put $n(i) = k + n + 2$ for all $i=\last(i)$.
Then, $\IP$ is:
\begin{eqnarray}
\nonumber
\hbox{minimize}  & \sum_{i=1}^n\sum_{j=\max\{i,k+1\}}^{n(i)-2} x_{i,j} & \\
\label{cons: input}
\hbox{s.t.} &\displaystyle{\sum_{j=\max\{i,k+1\}}^{k+N}x_{i,j} = 1} & \forall i\\
\label{cons: output}
                 &\displaystyle{\sum_{i=1}^{j}x_{i,j} = 1}          & \forall j\\
\label{cons: order}
                 &x_{n(i),j}-x_{i,j-1}\geq 0                                 & \forall i\neq \last(i);\ \forall j\geq n(i)\\
\label{cons: int}
                 & x_{i,j}\in \{0,1\}                                            & \forall i;  \forall j\geq i.
\end{eqnarray}
The constraints~\eqref{cons: input} guarantee that each item is eventually
removed from the buffer. The constraints~\eqref{cons: output} guarantee
that at each time slot one item is removed from the buffer.
The constraints~\eqref{cons: order} prevent the solution from switching
colors while there are still items of the current color in the buffer. These
constraints are needed to guarantee that the linear objective function
measures the cost of the solution correctly. Notice that the objective function
simply counts the number of items that are removed from the buffer before
the next item of the same color is encountered in the input. Without
constraints~\eqref{cons: order}, we could avoid paying for color changes
by keeping in the buffer just the last encountered item of a color $c$ until
the next item of this color is reached. We denote the optimal value of
$\IP$ by $z_{\IP}$.
\begin{proposition}\label{pr: IP = RBM}
The value of an optimal solution for ${\cal I}$ is exactly $z_{\IP}$.
\end{proposition}

\begin{proofsketch}
The obvious correspondence between RBM output sequences and feasible
$\IP$ solutions matches RBM solutions and $\IP$ solutions with the same
cost.
\end{proofsketch}

A linear programming relaxation $\LP$ is derived by relaxing the
constraints~\eqref{cons: int} to
\begin{eqnarray}
\label{cons: frac}
                & x_{i,j}\ge 0                                                     & \forall i;  \forall j\geq i.
\end{eqnarray}
We denote the value of the relaxation at a feasible solution $x$ by $z(x)$,
and the optimal value by $z_{\LP}$.
Given a feasible fractional solution $x$ of $\LP$ (i.e., $x$ satisfying
constraints~\eqref{cons: input}, \eqref{cons: output}, \eqref{cons: order},
and \eqref{cons: frac}) and a time step $j$, we can think of $x$ as defining
a fractional packing of input items into the buffer at time $j$. I.e., every
input item $i\le j$ is in the buffer with weight $w_i^j = w(x)_i^j$
where $w_i^j = 1 - \sum_{j'=i}^j x_{i,j'}$. For notational convenience, we
define $w_i^{i-1} = 1$. Notice that $w$ is a function of $x$; we usually
omit $x$ from the notation. Also notice that due to
constraints~\eqref{cons: input}, $w_i^j\in [0,1]$.
\begin{proposition}\label{pr: capacity constraint}
If $j\le n$, then $\sum_{i\le j} w_i^j = k$,
and if $j > n$ then $\sum_{i\le j} w_i^j = k + n - j$.
\end{proposition}

\begin{proofsketch}
By constraints~\eqref{cons: output}, for every $k+1\le j\le k+n$,
it holds that $\sum_{i\le j} x_{i,j} = 1$. Therefore, if $j\le n$,
$$
\sum_{i\le j} w_i^j  =  \sum_{i\le j} \left(1 - \sum_{j'=i}^j x_{i,j'}\right)
  =   j - \sum_{j'=k+1}^j \sum_{i\le j'} x_{i,j'}
  =  j - (j - k) = k.
$$
A similar argument shows the case of $j > n$.
\end{proofsketch}

\begin{proposition}\label{pr: monotonicity of w}
For every $i\ne\last(i)$, for every $j\ge n(i) - 1$, $w^j_i\le w^j_{n(i)}$.
\end{proposition}

\begin{proofsketch}
If at some point $j\ge n(i) - 1$ we have $w^j_i > w^j_{n(i)}$, then
$\sum_{j' > j} x_{i,j'} = w^j_i > w^j_{n(i)} \ge w^{j+1}_{n(i)} =
\sum_{j' > j+1} x_{n(i),j'}$, so at some point $j' > j+1$,
$x_{n(i),j'}-x_{i,j'-1} < 0$, violating constraints~\eqref{cons: order}.
\end{proofsketch}

Our main result is the following theorem.
\begin{theorem}\label{thm: main}
There is a constant $\alpha > 1$ and a (deterministic) polynomial time
algorithm which given a feasible fractional solution $x$ of $\LP$ computes
a feasible $0$-$1$ solution $\bar{x}$ of $\IP$ such that
$z(\bar{x})\le\alpha\cdot z(x)$.
\end{theorem}

\begin{corollary}\label{cor: main}
There is a (deterministic) polynomial time $\alpha$-approximation
algorithm for reordering buffer management.
\end{corollary}

\begin{proof}
Compute an optimal solution $x^*$ of $\LP$ with cost
$z(x^*) = z_{\LP}\le z_{\IP}$. Use Theorem~\ref{thm: main}
to compute a $0$-$1$ solution $\bar{x}^*$ with cost
$z(\bar{x}^*)\le\alpha z_{\LP}\le \alpha z_{\IP}$.
The corollary follows from Proposition~\ref{pr: IP = RBM}.
\end{proof}

We now describe the rounding algorithm of Theorem~\ref{thm: main}.
The algorithm works in phases. Each phase evicts one or more colors
from the buffer. To evict a color, the algorithm removes the items of
this color from the buffer until the buffer contains no such item. We
refer to the eviction of one color from the buffer as a step. The algorithm
uses constants $\delta_1,\delta_2,\delta_3\in (0,1)$, and
$\gamma=\gamma(\delta_1,\delta_2) > 1$, to be defined
later. In order to describe the algorithm, we need the following definition:
\begin{definition}
For $q=1,2,\dots,\lfloor z(x) / \delta_3 \rfloor$,
$$
t_q = \min\left\{t:\ \sum_{j=k+1}^t \sum_{i\le j} y_{i,j} \ge
    q\cdot\delta_3\right\},
$$
where $y_{i,j} = x_{i,j}$ if $n(i) > j + 1$, and $y_{i,j}=0$ otherwise. 
\end{definition}
In other words, $t_q$ denotes the earliest time at which the cost of
the fractional solution $x$ increases to at least $q\cdot\delta_3$.
The goal of phase $q$ is to reach $t_q$. (If the last $t_q$ was
already reached, the goal of the last phase is to reach the end
of the output sequence.) Each phase includes
three types of steps. There are steps that are charged against the
past $x$-cost of the items removed by that step. There are steps
that are paid for by charging the past $x$-cost of other items in
the buffer. Finally, there are a constant number of steps that cannot
be charged to the past $x$-cost, so they are charged to the increase
in $x$-cost that sets $t_q$.

For the second type of steps, we will use a charging scheme to determine
the choice of colors to evict, and also to analyze the algorithm. For each item
$i$ in the buffer at time $j$ we
maintain an index $\tau_{i,j}$ which is the last time before $j$ that $i$ was
charged. Initially, $\tau_{i,i} = i-1$. At time $j$, if $i$ was not charged at
time $j-1$, we put $\tau_{i,j} = \tau_{i,j-1}$. Otherwise, we put $\tau_{i,j} = j-1$.
(The charge is implied by the sequence of charging times.)
For each time step $j$, for each item $i$ in the buffer of the algorithm at time
$j$, define $d^j_i = w_i^{\tau_{i,j}} - w_i^j$. This is the fraction of $i$ that the
solution $x$ removed from the buffer since the last time $i$ was charged.

We call the set of all items of a specific color in the algorithm's
buffer a {\em block}. Let ${\cal B}^j = \{B^j_1,B^j_2,\dots,B^j_{m_j}\}$
denote the set of blocks in the algorithm's buffer at time $j$ (before
removing from the buffer an item at time $j$). For $r=1,2,\dots,m_j$,
let $f^j_r$ denote the earliest item in $B^j_r$.
For an item $i$, we denote $t(i) = \min\{t: w_i^t\le 1 - \delta_1\}$.
We assume that the
blocks are ordered so that $t(f^j_1)\le t(f^j_2)\le \cdots\le t(f^j_{m_j})$.
We denote by $\Delta_j$ the difference in volume between
the algorithm's buffer and the fractional buffer at time $j$.
Formally,
$$
\Delta_j = \frac 1 2 \|\bar{w}^j - w^j\|_1 =
\sum_{r=1}^{m_j}\sum_{i\in B^j_r} (1 - w^j_i) =
\sum_{i \le j:\ \bar{w}^j_i = 0} w^j_i,
$$
where $\bar{w} = w(\bar{x})$.
For a current time step $j$ and a target time step $t_q$, let
$$
I^j_q = \argmax\{|I|:\ I\subset [j,t_q]\wedge \forall i,i'\in I, c(i) = c(i')\},
$$
and put
$$
t^j_q = \max\{j, t_q + 1 - |I^j_q|\}.
$$
This is the earliest time where items from one color that arrive after
time $j-1$ can be removed from the buffer consecutively,
reaching time $t_q$ or later. (Notice that if $t_q < j$, then $t^j_q = j$.)
Intuitively, if we reach $t^j_q$ without using items from $I^j_q$,
then in one more step we can reach our target $t_q$.
Consider a decision time $j$ (the previous phase ended at time $j-1$).
We execute the following procedure:
\begin{enumerate}
\setcounter{enumi}{-1}
\item\label{it: rule 0} While our buffer contains an item $i$ such that
         $t(i)\le j$, we evict color $c(i)$, and we increment $j$
         to be the first time step following the one we've reached. When
         there are no more steps of this case, we execute the first case
         among~\ref{it: rule 1}--\ref{it: rule 4} that applies.
\item\label{it: rule 1} If there is a color $c$ that we can evict and
         reach $t_q$, we evict one such color, thus ending the phase.
\item\label{it: rule 2} If our buffer contains $t^j_q - j$ items from one
         or two colors, we evict those colors in two steps, and if we haven't
         reached $t_q$, we also evict the color of $I^j_q$. We
         will prove in Claim~\ref{cl: rule 2} that we reach $t_q$, so the phase
         ends.
\item\label{it: rule 3} If $\Delta_j\ge\frac{1}{\gamma}\cdot (t^j_q - j)$,
         then we do the following. For $B\in {\cal B}^j$, put
         $\hat{d}^j_B = \frac{1}{|B|}\cdot\sum_{i\in B} d^j_i$. Define
         $s_1 > s_2 > \cdots > s_{p_j}$ inductively as follows. Initially
         set $s_1 = m_j$. Assuming $s_p$ is defined, let $r_p\le s_p$ be the
         largest index for which
         $\sum_{u=r_p}^{s_p-1} \sum_{i\in B^j_u} d^j_i\le \delta_2 |B^j_{s_p}|$.
         If $\sum_{u=r_p}^{s_p} \hat{d}^j_{B^j_u} \ge \delta_1$,
         let $r'_p\in [r_p,s_p]$ be the largest index for which
         $\sum_{u=r'_p}^{s_p} \hat{d}^j_{B^j_u} \ge \delta_1$.
         (See Figure~\ref{fig.rule3} in the appendix.)
         We evict the color of block $B^j_{s_p}$, and charge the items in
         $B^j_{r'_p},\dots,B^j_{s_p-1}$
         (i.e., for each charged item $i$, set $\tau_{i,j+1} = j$). If $r'_p > 1$,
         set $s_{p+1}$ to be $r'_p-1$, else set $p_j = p$.
         Otherwise, if $\sum_{u=r_p}^{s_p} \hat{d}^j_{B^j_u} < \delta_1$
         and $r_p > 1$, set $s_{p+1} = \argmax\{|B^j_u|:\ u\in [r_p-1,s_p-1]\}$. 
         Otherwise, if $r_p = 1$, set $p_j = p$.
         If the entire process removes fewer than $t^j_q - j$ 
         items that were in our buffer at time $j$,
         we evict the color of the largest block $B\in {\cal B}^j$
         that remains. We will prove in Claims~\ref{cl: rule3_a} and~\ref{cl: rule3_b}
         that at this point we can evict the color of $I^j_q$ and reach $t_q$, thus
         ending the phase.
\item\label{it: rule 4} In the remaining case,
         we evict the color of the largest block $B\in {\cal B}^j$.
         We increment $j$
         to be the time step following the last output step. Now, we
         execute the procedure again. We will prove in Claim~\ref{cl: rule 4}
         that in a phase we never reach case~\ref{it: rule 4} twice.
\end{enumerate}
This completes the definition of $\bar{x}$.

\section{Analysis}

In this section we prove our main result, Theorem~\ref{thm: main}.
We first give an interpretation of the feasible fractional solution $x$.
Consider a color $c$, a sequence $I$ of color $c$ items $i_1,i_2,\dots,i_m$
and a starting time $j$. Let $M_{I,j}$ denote the matching given by
$M_{I,j}(i_s) = j+s$, for all $s=1,2,\ldots,m$.
We say that $M_{I,j}$ is a {\em monochromatic sequence matching} (MSM)
iff the items are a maximal sequence of consecutive items of the same color $c(i_1)$.
In other words, $M_{I,j}$ is an MSM iff it
satisfies the following conditions:
($i$) for every $s=1,2,\ldots,m-1$ it holds that $c(i_s)=c(i_{s+1})$
and $n(i_{s})=i_{s+1}$;
($ii$) $j+s\geq i_s$ for every $s=1,2,\ldots,m$;
($iii$) $j+m< n(i_m)-1$.
\begin{proposition}\label{pr: packing sequences}
For every feasible fractional solution $x$ there is a fractional packing
of monochromatic sequence matchings $\lambda = \lambda(x)$
that satisfies the following constraints:
(a) for each input item $i$, $\sum_{I,j: i\in I} \lambda_{I,j} = 1$;
(b) for each time slot $t$, $\sum_{I,j: t \in [j+1, j+|I|]} \lambda_{I,j} = 1$;
(c) $z(x) = \sum_{I,j} \lambda_{I,j}$.
\end{proposition}

\begin{proofsketch}
We can construct $\lambda$ by the following algorithm.
While there exist $i_1,j$ such that $x_{i_1,j+1} > 0$, find
such a pair with minimum $j$. Find a maximal sequence
$I = (i_1,i_2,\dots,i_m)$ of items of color $c(i_1)$ with
$x_{i_s,j+s} > 0$. By constraints~\eqref{cons: order}
(which are maintained through the induction), it must
be that $x_{i_s,j+s} \ge x_{i_1,j+1}$ and $j+m < n(i_m) - 1$.
Put $\lambda_{I,j} = x_{i_1,j+1}$ and subtract $\lambda_{I,j}$
from $x_{i_s,j+s}$ for all $s=1,2,\dots,m$. (Notice that this
will not cause constraints~\eqref{cons: order} to be violated.)
Constraints (a), (b) follow from constraints \eqref{cons: input},
\eqref{cons: output} of the $\LP$. Equation (c) follows from
the fact that every MSM $M_{I,j}$ that we construct ends at
time $j+|I|$ which precedes the arrival of the next item of
this color.
\end{proofsketch}

We next prove that the algorithm is well-defined.
\begin{claim}\label{cl: rule 2}
Executing case~\ref{it: rule 2} ends a phase.
\end{claim}

\begin{proof}
Consider a phase $q$ where we execute case~\ref{it: rule 2}
at time $j$.
Let $B,B'\in {\cal B}^j$ denote the two blocks with
$|B| + |B'|\ge t^j_q - j$, and let $c,c'$ denote their
colors. (If the $t^j_q - j$ items stipulated by
case~\ref{it: rule 2} are of a single color, set
$B' = \emptyset$ and $c' = c$.)
Recall that $I^j_q$ is the set of items that determine
$t^j_q$. Notice that
$I^j_q\cap (B\cup B') = \emptyset$, because all
the items in $B\cup B'$ arrived before time $j$,
and all the items in $I^j_q$ arrive at time $j$ or
later. Let $b$ denote the number of items in
$I^j_q$ that are removed when we evict the colors
$c,c'$. Let $I'$ denote the set of remaining items
from $I^j_q$. Notice that $b > 0$ only if the color
of $I^j_q$ is either $c$ or $c'$. When we evict the
colors of $B,B'$, we reach $t^j_q+b-1$. As the items
in $I'$ can be removed starting from time $t^j_q+b$
and ending at time $t_q$, evicting the color of
$I'$ ends the phase.
\end{proof}

\begin{claim}\label{cl: rule3_a}
For every $\delta_1,\delta_2 > 0$ such that $\delta_2 > 2\delta_1$
there exists $\gamma =\gamma(\delta_1,\delta_2)$ such that
applying the process in case~\ref{it: rule 3} starting at time $j$
removes at least $t^j_q-j$ items that were in the buffer at time $j$.
\end{claim}

\begin{proof}
Let $\Delta_j^F = \sum_{B\in {\cal B}^j}\sum_{i\in B} d^j_i$ be
the uncharged portion of items that are removed by $x$, but the algorithm holds
at time $j$. We start by showing that $\Delta_j^F \geq
\frac{1-\delta_1-\delta_2}{1-\delta_1} \cdot \Delta_j$.
Each block $B$ that the algorithm removed from the buffer
before time $j$ contributes $\sum_{i\in B} w^j_i$ to $\Delta_j$
(and the sum of all those contributions is exactly $\Delta_j$).
Whenever the algorithm removes a block $B$,
it charges some of the volume of
the items that remain in its buffer and paid for removing $B$.
The total volume charged is at most $\delta_2\cdot |B|$.
(This is trivially true when $B$ is not removed during a
case~\ref{it: rule 3} process and does not charge anything.)
If all the items that are charged when $B$ is removed
are not in the algorithm's buffer at time $j$,
the same contribution of $\sum_{i\in B} w^j_i$ contributes
to $\Delta_j^F$ as well. Now consider the case that the
buffer does contain items that were charged
for the removal of $B$. Let $j'$ be the beginning
of the case~\ref{it: rule 3} process in which
block $B = B^{j'}_s$ was removed,
and let $B^{j'}_{p}$ be one of the blocks that
were charged for removing $B$ and its items
are still in the algorithm's buffer at time $j$.
By the ordering of the blocks in ${\cal B}^{j'}$, it must
be that $t(f^{j'}_{p}) \leq t(f^{j'}_{s})$. As we can't apply
case~\ref{it: rule 0}, $w^j_{f^{j'}_{p}}\geq 1-\delta_1$.
Therefore, the first item $f^{j'}_{s}$ of block $B$ also has
$w^j_{f^{j'}_{s}}\geq 1-\delta_1$.
As every item in block $B$ must have at least the same weight
as the first item,
we get that $\sum_{i\in B} w^j_i
\geq (1-\delta_1)\cdot|B|$. Block $B$ only charges at most a
volume of $\delta_2\cdot |B|$ of items that are still in
the algorithm's buffer at time $j$.
Putting $\sum_{i\in B} w^j_i = (1-\theta)\cdot|B|$, we get that $B$
contributes to $\Delta_j^F$ at least
$$\frac {(1-\theta)\cdot|B|-\delta_2\cdot |B|}{(1-\theta)\cdot|B|}
\geq \frac {1-\theta-\delta_2 - (\delta_1-\theta)}{1-\theta- (\delta_1-\theta)}
= \frac {1-\delta_1-\delta_2}{1-\delta_1}
$$
of the portion it contributes to $\Delta_j$.

Going back to the main argument,
Let $s_{e(1)} > s_{e(2)}>\dots>s_{e(\ell)}$ denote the indices of the blocks
that we removed during the case~\ref{it: rule 3} process
($\{s_{e(1)},\dots,s_{e(\ell)}\}\subseteq\{s_1, s_2,\dots,s_{p_j}\}$).
For each $p\in \{1,\dots,\ell\}$, by the definition of $r'_{e(p)}$,
$\sum_{u=r'_{e(p)}}^{s_{e(p)}-1} \sum_{i\in B^j_u} d^j_i\le
\delta_2\cdot |B^j_{s_{e(p)}}|$,
and $\sum_{u=r'_{e(p)}}^{s_{e(p)}} \hat{d}^j_{B^j_u} \ge \delta_1$.
Consider the indices $s_{e(p)+1},s_{e(p)+2},\ldots,s_{e(p+1)-1}$. (Those
are the indices defined in the process of blocks that weren't removed
between removing block $B^j_{s_{e(p)}}$ and block $B^j_{s_{e(p+1)}}$.)
We now show that for each $u \in [e(p)+1,e(p+1)-1]$,
 $|B_{s_{u}}^j|< \frac{2\delta_1}{\delta_2}\cdot|B^j_{s_{u+1}}|$
(the same holds for $u\in[1,e(1)-1]$ and $u\in[e(\ell)+1,p_j-1]$):
\begin{eqnarray*}
\delta_2\cdot |B_{s_u}^j| &<& \sum_{g=r_u-1}^{s_u-1} \sum_{i \in B^j_g} d^j_i =
\sum_{g=r_u-1}^{s_u-1} \hat{d}^j_{B^j_{g}} \cdot|B^j_{g}|
\le \sum_{g=r_u-1}^{s_u-1} \hat{d}^j_{B^j_g} \cdot|B^j_{s_{u+1}}| \\
&=&|B^j_{s_{u+1}}|\cdot \left(\hat{d}^j_{B^j_{r_u-1}} +
\sum_{g=r_u}^{s_u-1} \hat{d}^j_{B^j_g}\right)
<  2\delta_1 \cdot |B^j_{s_{u+1}}|.
\end{eqnarray*}
The first inequality follows from the definition of $r_u$.
The second inequality follows as $B^j_{s_{u}}$ is not removed by
the process, therefore, $B^j_{s_{u+1}}$ is defined to be the
maximal block in $\{B^j_{r_u-1},B^j_{r_u},\ldots, B^j_{s_u-1}\} $.
The last inequality follows as
$\sum_{u'=r_u}^{s_u-1} \hat{d}^j_{B^j_{u'}} < \delta_1$.
Furthermore, as the algorithm did not continue to execute case~\ref{it: rule 0}
at time $j$,
$d^j_{i} \leq 1-w^j_{i}<\delta_1$ for every $i \in B^j_{r_u-1}$,
and therefore the average over $i \in B^j_{r_u-1}$ of $d^j_i$ is also
less than $\delta_1$.
Thus we get $|B_{s_{u}}^j|<
\left(\frac{2\delta_1}{\delta_2}\right)^{e(p+1)-u} \cdot|B^j_{s_{e(p+1)}}|$.
We can now bound the contribution to $\Delta_j^F$ of the blocks
with indices in $[s_{e(p+1)}+1,s_{e(p)+1}]$
(blocks that weren't removed and weren't
charged for removing any block) as follows:
\begin{eqnarray*}
\sum_{g=s_{e(p+1)}+1}^{s_{e(p)+1}} \sum _{i\in B_g^j} d_i^j &=&
\sum_ {u=e(p)+1}^{e(p+1)-1} \sum_{g=s_{u+1}+1}^{s_u} \sum _{i\in B_g^j} d_i^j
\le  \sum_ {u=e(p)+1}^{e(p+1)-1} (\delta_2+\delta_1)\cdot |B^j_{s_u}|\\
&\le &  (\delta_2+\delta_1)\cdot \sum_ {u=e(p)+1}^{e(p+1)-1}
\left(\frac{2\delta_1}{\delta_2}\right)^{e(p+1)-u}\cdot|B^j_{s_{e(p+1)}}|\\
&\le & (\delta_2+\delta_1) \cdot |B^j_{s_{e(p+1)}}| \cdot
\frac{2\delta_1}{\delta_2- 2\delta_1}
= \frac{2\delta_1(\delta_2+\delta_1)}{\delta_2- 2\delta_1}  \cdot |B^j_{s_{e(p+1)}}|,
\end{eqnarray*}
Adding the contributions of block $B^j_{s_{e(p+1)}}$
and the blocks that were charged for its removal, we get:
$$
\sum_{g=r'_{e(p+1)}}^{s_{e(p)+1}}\sum _{i\in B_g^j} d_i^j
\le \left(\delta_2+\delta_1+\frac{2\delta_1(\delta_2+\delta_1)}{\delta_2- 2\delta_1}\right)
\cdot |B^j_{s_{e(p+1)}}| =
\frac{\delta_2(\delta_2+\delta_1)}{\delta_2- 2\delta_1}
\cdot |B^j_{s_{e(p+1)}}|
$$
For the same reason,
$$\sum_{g=r'_{e(1)}}^{s_1=m(j)} \sum _{i\in B_g^j} d_i^j\le
\frac{\delta_2(\delta_2+\delta_1)}{\delta_2- 2\delta_1}
\cdot |B^j_{s_{e(1)}}|$$ and
$$ \sum_{g=r_{p_j}=1}^{s_{e(\ell)+1}} \sum _{i\in B_g^j} d_i^j
\le \frac{\delta_2(\delta_2+\delta_1)}{\delta_2- 2\delta_1}
 \cdot |B^j_{s_{p_j}}|.$$
Therefore, if $p_j > e(\ell)$ then
\begin{eqnarray}\label{eq: delta_j^F}
\Delta_j^F&=& \sum_{g= 1}^{m(j)} \sum_{i\in B^j_g} d_i^j
= \sum_{g=r_{p_j}=1}^{s_{e(\ell)+1}} \sum _{i\in B_g^j} d_i^j
+\sum_{p=1}^{\ell-1} \sum_{g=r'_{e(p+1)}}^{s_{e(p)+1}}\sum _{i\in B_g^j} d_i^j
+\sum_{g=r'_{e(1)}}^{s_1=m(j)} \sum _{i\in B_g^j} d_i^j \nonumber\\
&\le& \frac{\delta_2(\delta_2+\delta_1)}{\delta_2- 2\delta_1}
\cdot\left(|B^j_{s_{p_j}}| +\sum_{p=1}^{\ell}  |B^j_{s_{e(p)}}|\right).
\end{eqnarray}
Thus,
\begin{eqnarray}\label{eq: number of items}
|B^j_{s_{p_j}}| +\sum_{p=1}^{\ell}  |B^j_{s_{e(p)}}|&\ge&
\frac{\delta_2- 2\delta_1}{\delta_2(\delta_2+\delta_1)}\cdot \Delta_j^F
\ge \frac{\delta_2- 2\delta_1}{\delta_2(\delta_2+\delta_1)}
\cdot \frac{1-\delta_1-\delta_2}{1-\delta_1} \cdot \Delta_j \nonumber\\
&\ge& \frac{\delta_2- 2\delta_1}{\delta_2(\delta_2+\delta_1)}
\cdot \frac{1-\delta_1-\delta_2}{1-\delta_1} \cdot\frac{1}{\gamma}\cdot (t_q^j-j).
\end{eqnarray}
Choosing $\gamma = \frac{\delta_2- 2\delta_1}{\delta_2(\delta_2+\delta_1)}
\cdot \frac{1-\delta_1-\delta_2}{1-\delta_1}$ we get that
$|B^j_{s_{p_j}}| +\sum_{p=1}^{\ell}  |B^j_{s_{e(p)}}|$ is at least $t_q^j-j$.

To conclude, by the definition of the process we removed from
the buffer at least $\sum_{p=1}^{\ell}  |B^j_{s_{e(p)}}|$
items from the items that were in the buffer at time $j$.
Notice that if $p_j =e(\ell)$ we get that, similar to Equation~\eqref{eq: delta_j^F},
$\Delta_j^F \le
\frac{\delta_2(\delta_2+\delta_1)}{\delta_2- 2\delta_1}
\cdot \sum_{p=1}^{\ell}  |B^j_{s_{e(p)}}|$,
and for the same reason as in Equation~\eqref{eq: number of items},
$\sum_{p=1}^{\ell} |B^j_{s_{e(p)}}| > t_q^j-j$.
Otherwise, it must be that the last block $B_{s_{p_j}}^j$ considered in the process
was not removed. Thus, removing the largest block $B\in {\cal B}^j$ that
remained will remove at least $|B^j_{s_{p_j}}|$ items, and overall
at least $|B^j_{s_{p_j}}| +\sum_{p=1}^{\ell}  |B^j_{s_{e(p)}}|\ge  t_q^j-j$
items.
\end{proof}

\begin{claim}\label{cl: rule3_b}
Executing case~\ref{it: rule 3} ends a phase.
\end{claim}

\begin{proof}
Consider phase $q$ where we execute case~\ref{it: rule 3}
at time $j$.
Let $j'$ be the time we execute the last step of
case~\ref{it: rule 3},
and let $I$ be the set of items from $I^j_q$ that were removed
before time $j'$. By Claim~\ref{cl: rule3_a}, at least $t_q^{j}-j$
removed items were in the buffer at time $j$,
therefore at least $t_q^{j}-j+|I|$ items were removed overall.
Thus, $j' \ge t_q^j + |I|$. We can now evict the color of $I^j_q$
and reach $t_q$, as there are at least
$|I^j_q|-|I|= t_q - t_q^{j} +1 - |I| \ge t_q - j' + 1$ items
of this color that can be removed from the buffer consecutively
starting at time $j'$.
\end{proof}

Let $\phi = \frac{1 + \sqrt{5}}{2}\approx 1.618$ denote
the golden ratio.
\begin{claim}\label{cl: rule 4}
Assuming that $\gamma > \frac{1+\phi}{1 - \delta_3}$,
if we've reached case~\ref{it: rule 4}, then in the repeated
execution of the procedure we execute one of the
cases~\ref{it: rule 1}--\ref{it: rule 3}.
\end{claim}

\begin{proof}
Suppose that case~\ref{it: rule 4} is executed at time $j$
in phase $q$.
We may assume that $t_q > j$, otherwise the claim is
vacuous. Let
$$
{\cal L} = \{(I,j'): j' < j\wedge j'+|I|\ge t_q
\wedge \lambda_{I,j'} > 0\}
$$
denote the set of monochromatic sequences
in the packing $\lambda$ that are matched
to an interval containing the entire interval
$[j,t_q]$. (See Figure~\ref{fig.defs} in the
appendix.)
Notice that by the definition of $t_q$ and the fact
that $t_{q-1} < j < t_q$,
\begin{equation}\label{eq: long seq weight}
\Lambda = \sum_{(I,j')\in{\cal L}} \lambda_{I,j'}\ge 1 - \delta_3.
\end{equation}
By the definition of $t^j_q$, for every $(I,j')\in {\cal L}$,
none of the items in $M^{-1}_{I,j'}([j,t^j_q-1])$ (the items
matched by $M_{I,j'}$ to the interval $[j,t^j_q-1]$) arrive
at time $j$ or later. As we've reached case~\ref{it: rule 4},
we may conclude that $\max\{|B|:\ B\in {\cal B}^j\} < t^j_q-j$
(otherwise case~\ref{it: rule 2} would apply)
and $\Delta_j < \frac 1 \gamma\cdot (t^j_q-j)$
(otherwise case~\ref{it: rule 3} would apply).
In particular, consider $(I,j')\in {\cal L}$.
Let $t_{I,j'}$ denote the minimum time $t'$
for which $M^{-1}_{I,j'}(t')$ is in the algorithm's
buffer. If no such time exists, set $t_{I,j'} = t^j_q$.
The items in $M^{-1}_{I,j'}([j,t_{I,j'}-1])$ are no longer
in the algorithm's buffer at time $j$. Therefore,
$\sum_{(I,j')\in {\cal L}} \lambda_{I,j'} (t_{I,j'} - j)
\le\Delta_j$. Using Equation~\eqref{eq: long seq weight},
we conclude that
$E[t_{I,j'} - j] < \frac{1}{\gamma (1 - \delta_3)} (t^j_q - j)$,
where the expectation is taken over $(I,j')\in {\cal L}$
with probability distribution
$\Pr[(I,j')] = \frac{\lambda_{I,j'}}{\Lambda}$.
(See Figure~\ref{fig.defs} in the appendix.)
In particular,
\begin{equation}\label{eq: longest}
\min\{t_{I,j'} - j:\ (I,j')\in {\cal L}\} < \frac{1}{\gamma (1 - \delta_3)} (t^j_q - j),
\end{equation}
and by Markov's inequality
\begin{equation}\label{eq: weight of long}
\Pr\left[t_{I,j'} - j < \frac{\phi}{\gamma (1 - \delta_3)} (t^j_q - j)\right] \ge
1 - \frac{1}{\phi} = \frac{\phi-1}{\phi} = \frac{1}{1+\phi}.
\end{equation}
Let
$$
(I_{\min},j'_{\min}) = \argmin\{t_{I,j'} - j:\ (I,j')\in {\cal L}\}.
$$
Let $B\in {\cal B}^j$ denote the block containing items from
$I_{\min}$ in the algorithm's buffer.
Notice that by Equation~\eqref{eq: longest},
$$
|B| > \left(1 - \frac{1}{\gamma (1 - \delta_3)}\right)\cdot (t^j_q-j).
$$
As we've reached case~\ref{it: rule 4} (and therefore case~\ref{it: rule 2}
does not apply), for all other blocks $B'\in {\cal B}^j$,
$$
|B'| < t^j_q - j - |B|\le \frac{1}{\gamma (1 - \delta_3)} (t^j_q - j)
< \left(1 - \frac{1}{\gamma (1 - \delta_3)}\right)\cdot (t^j_q - j),
$$
where the last inequality uses
$\gamma > \frac{1+\phi}{1 - \delta_3} > \frac{2}{1 - \delta_3}$.
Thus, the case~\ref{it: rule 4} step at $j$ must evict the color of
$B$. As case~\ref{it: rule 1} did not apply, evicting the color
of $B$ does not reach $t_q$. Denote
$$
{\cal L}' = \left\{(I,j')\in {\cal L}:\
   t_{I,j'} - j < \frac{\phi}{\gamma (1 - \delta_3)} (t^j_q - j)\right\}.
$$
For every $(I,j')\in {\cal L}'$ it must hold that the portion of
$I$ in the algorithm's buffer is $B$. This is because for all
other colors
$t_{I,j'} - j > \left(1 - \frac{1}{\gamma (1 - \delta_3)}\right) (t^j_q - j)$,
and as $\gamma > \frac{1+\phi}{1 - \delta_3}$, we get
$t_{I,j'} - j > \frac{\phi}{\gamma (1 - \delta_3)} (t^j_q - j)$.
So consider the situation after the step at $j$, where
we remove $B$ and possibly additional items of $B$'s color,
and we reach $t' < t_q$. Consider the set $A$ of the last
$t^j_q - j - |B|$ items that the algorithm removed so far.
For all $i\in A$ and for all $(I,j')\in {\cal L}'$, $M_{I,j'}(i) > t'$.
Therefore, by Equation~\eqref{eq: weight of long},
$w_i^{t'+1}\ge \frac{1}{1+\phi} (1 - \delta_3) > \frac{1}{\gamma}$.
On the other hand, $\bar{w}_i^{t'+1} = 0$.
Let $t''$ be the point matched by $M_{I_{\min},j'_{\min}}$ to
the first item $i'$ of $I_{\min}$ that wasn't yet encountered.
Notice that for all $(I,j')\in {\cal L}'$, $M_{I,j'}(i')\ge t''$.
Clearly, $t''\ge t^{t'+1}_q$. Therefore,
$$
\Delta_{t'+1} \ge\sum_{i\in A} w_i^{t'+1} >
\frac{1}{\gamma}\cdot (t'' - t' - 1)\ge
\frac{1}{\gamma}\cdot (t^{t'+1}_q - t' - 1).
$$
When we execute the procedure again, we first execute
case~\ref{it: rule 0}. Assuming that we haven't reached $t_q$,
the following holds. Each removed item moves our current
position $t'$ by $1$, and we may lose $\frac{1}{\gamma}$
in our estimate of $\Delta_{t'+1}$ for each increment of $t'$.
Each removed item with the color of $B$ moves the target $t''$
by $1$, but in those steps we do not lose $\frac{1}{\gamma}$
in our estimate of $\Delta_{t'+1}$. (Notice that if we do not
reach $t_q$, then we haven't yet encountered this item in any
of the sequences in ${\cal L}'$ and the reason for its removal
must be other sequences.) Thus, with respect to the new $t'$,
we still have that
$\Delta_{t'+1}\ge \frac{1}{\gamma}\cdot (t^{t'+1}_q - t' - 1)$,
so if cases~\ref{it: rule 1} and~\ref{it: rule 2} do not apply,
then case~\ref{it: rule 3} applies.
\end{proof}

We now analyze the charging scheme that is used in
case~\ref{it: rule 3}. We say that a block $B^j_u\in {\cal B}^j$
that is charged at time $j$ pays $\hat{d}^j_{B^j_u}$ (towards
evicting the color of $B^j_{s_p}$ for which $r'_p\le u < s_p$).
Denote by $\widehat{\cal B}^j$ the set of blocks that are charged
at time $j$.
\begin{lemma}\label{lm: charge}
$\sum_j \sum_{B\in \widehat{\cal B}^j} \hat{d}^j_B\le 2\cdot z(x)$.
\end{lemma}

\begin{proof}
Fix $j$ and consider a block $B\in\widehat{\cal B}^j$ which pays
$\hat{d}^j_{B}$ at time $j$. Notice that
$$
\hat{d}^j_{B} = \frac{1}{|B|}\cdot\sum_{i\in B} d^j_i\le\max\{d^j_i:\ i\in B\}.
$$
Let $i_B = \argmax\{d^j_i:\ i\in B\}$. Then, $d^j_{i_B}$ is simply the
sum of $\lambda_{I,j'}$ over $(I,j')$ such that
$M_{I,j'}(i_B)\in (\tau_{i_B,j},j]$. So we can think of such $(I,j')$
as contributing $\lambda_{I,j'}$ towards the payment of $\hat{d}^j_{B}$.
If $j' + |I|\le j$, then $(I,j')$ will never be ``asked" to contribute again.
However, if $j' + |I| > j$, then $I$ contains items that are not charged
at time $j$, and such items may appear in a future block $B'$ that
pays for removing some block in a future phase. We argue that
it must be the case that $j' + |I| < t_q$. To prove that, assume for
contradiction that this is not the case. Notice that at time $j$ the item
$i_B$ is still in the algorithm's buffer, and $M_{I,j'}(i_B)\le j$. So
if the algorithm evicts $c(i_B)$, it must reach $t_q$. This contradicts
the assumption that at time $j$ the algorithm executes case~\ref{it: rule 3},
as case~\ref{it: rule 1} applies. Therefore, if $(I,j')$ contributes again
in a future phase at some time $j''$, we have that $j' + |I| < t_q < j''$,
so $(I,j')$ will never contribute a third time.
\end{proof}

We are now ready to prove our main result.
\begin{proofof}{Theorem~\ref{thm: main}}
We choose $\delta_1,\delta_2,\delta_3\in (0,1)$
such that $\delta_2 > 2\delta_1$ and
$\gamma = \frac{1-\delta_1-\delta_2}{1-\delta_1}\cdot
\frac{\delta_2- 2\delta_1}{\delta_2(\delta_2+\delta_1)}
> \frac{1 + \phi}{1-\delta_3}$.
(For example, we can choose $\delta_1 = \frac{1}{40}$,
$\delta_2 = \frac{1}{10}$, and $\delta_3 = \frac{1}{5}$.)

The number of phases is at most $\lceil z(x) / \delta_3\rceil$.
In a phase, cases~\ref{it: rule 1}--\ref{it: rule 4} are executed
at most once. The total number of steps due to case~\ref{it: rule 1},
case~\ref{it: rule 2}, case~\ref{it: rule 4}, and the last two steps of
case~\ref{it: rule 3} is at most 4. (The worst case is when the
algorithm executes case~\ref{it: rule 4} and then case~\ref{it: rule 2}.)
Therefore, the total cost of those steps is at most
$\frac{4}{\delta_3}\cdot z(x) + 4$.

Now consider case~\ref{it: rule 0}. When a color $c$ is evicted
at time $j$, this is because there's a block $B$ in the algorithm's
buffer with $i\in B$ such that $w^j_i\le 1 - \delta_1$. In particular,
$w^j_{f^j_B}\le 1 - \delta_1$. We remove $f^j_B$ at time $j$, so
for every monochromatic sequence $(I,j')$ with $M_{I,j'}(f^j_B)\le j$,
when this step is over we've reached $j'+ |I|$. Those are the sequences
that pay for the drop of at least $\delta_1$ by time $j$ in the weight
of $f^j_B$. As we've reached past them, we will never count them
again for another step of case~\ref{it: rule 0}. Thus, the total number
of such steps is at most $\frac{1}{\delta_1}\cdot z(x)$.

The remaining steps are color evictions due to the charging
scheme of case~\ref{it: rule 3}. Notice that whenever a block
$B_{s_p}\in {\cal B}^j$ is removed by this case, we have that
$\sum_{u=r'_p}^{s_p} \hat{d}^j_{B^j_u}\ge \delta_1$.
Therefore, the number of such steps is at most
$\frac{1}{\delta_1}\cdot \sum_j \sum_{p\in I^j_e}
\sum_{u=r'_p}^{s_p} \hat{d}^j_{B^j_u}$,
where $I^j_e$ denotes the set of indices of blocks whose removal
created a charge at time $j$. By Lemma~\ref{lm: charge},
$\frac{1}{\delta_1}\cdot
\sum_j \sum_{p\in I^j_e}\sum_{u=r'_p}^{s_p} \hat{d}^j_{B^j_u}\le
\frac{2}{\delta_1}\cdot z(x)$.
The total cost of all the cases is
$z(\bar{x})\le \left(\frac{3}{\delta_1} + \frac{4}{\delta_3}\right)\cdot z(x) + 4$.
As $z(x)\ge |C|\ge 1$, we get that the approximation guarantee $\alpha$
satisfies
$\alpha\le \frac{3}{\delta_1} + \frac{4}{\delta_3} + 4$.
\end{proofof}

\section{Concluding Remarks}

Our methods can be adapted easily to handle some
additional constraints, such as incurring a color
change cost whenever we've accummulated too
many time steps without a color change (a constraint
relevant to~\cite{ROADEF2005}).
Numerical estimates of the best constant that the above
analysis gives indicate that it is below $135$, taking
$\delta_1\approx 0.02763$, $\delta_2\approx 0.11416$,
and $\delta_3\approx 0.18481$. We did not
attempt to optimize our analysis, however, it is unlikely
that our methods can be pushed to yield a very small constant
(such as $2$). Substantially improving the approximation
guarantee for RBM is an interesting open problem. Also,
adapting our methods to deal with more general cost measures
appears to be a non-trivial task. In a variant of RBM called the
(uniform) $k$-client problem~\cite{ATUW01}, instead of a buffer
there are $k$ input sequences. At each time step, the next item
from one of the sequences is chosen and moved to the output
sequence. (So the choice of which item to remove affects the order
of the combined input sequence.) The goal is the same as RBM: to
minimize the number of color changes in the output sequence.
Adapting our methods to deal with this setting seems to be another
fascinating problem.

\bigskip\bigskip\bigskip

\appendix

\section*{Appendix: Figures}

\begin{figure}[ht]
\center
\includegraphics[scale=.5]{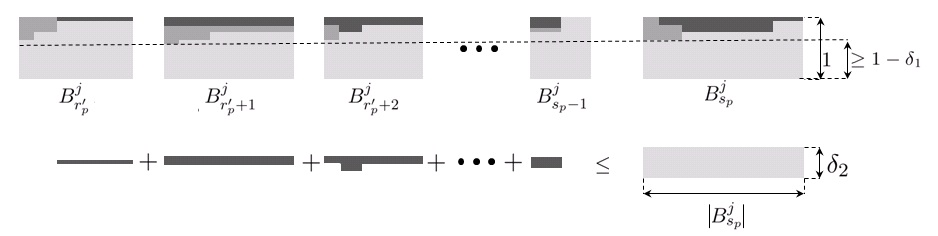}
\caption{The black area indicates the $d_i^j$-s, the dark grey area indicates
the remaining portion of the $x_{i,j}$-s that was previously charged, and
the light grey area indicates the $w_i^j$-s.}
\label{fig.rule3}
\end{figure}

\bigskip\bigskip\bigskip

\begin{figure}[ht]
\center
\includegraphics[scale=.5]{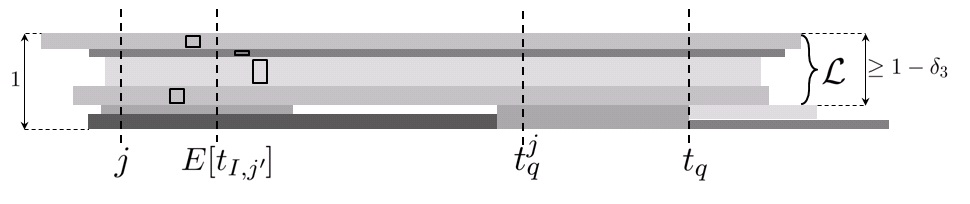}
\caption{The strips of varying tones indicate the packing of MSMs
in the fractional solution, and the outlined rectangles indicate the
$t_{I,j'}$-s.}
\label{fig.defs}
\end{figure}


\newpage

\begin{thebibliography}{10}

\bibitem{Abo08}
A.~Aboud.
\newblock Correlation clustering with penalties and approximating the
  reordering buffer management problem.
\newblock Master's thesis, Computer Science Department, The {Technion} -
  {Israel} Institute of Technology, January 2008.

\bibitem{ACER11}
A.~Adamaszek, A.~Czumaj, M.~Englert, and H.~R\"acke.
\newblock Almost tight bounds for reordering buffer management.
\newblock In {\em Proc. of the 43rd Ann. ACM Symp. on Theory of Computing},
  pages 607--616, June 2011.

\bibitem{Alb04}
S.~Albers.
\newblock New results on web caching with request reordering.
\newblock In {\em Proc. of the 16th ACM Symp. on Parallel Algorithms and
  Architectures}, pages 84--92, 2004.

\bibitem{ATUW01}
H.~Alborzi, E.~Torng, P.~Uthaisombut, and S.~Wagner.
\newblock The k-client problem.
\newblock {\em J. Algorithms}, 41(2):115--173, 2001.

\bibitem{AKM10}
Y.~Asahiro, K.~Kawahara, and E.~Miyano.
\newblock NP-hardness of the sorting buffer problem on the uniform metric.
\newblock Unpublished, 2010.

\bibitem{AR10}
N.~Avigdor-Elgrabli and Y.~Rabani.
\newblock An improved competitive algorithm for reordering buffer management.
\newblock In {\em Proc. of the 21st Ann. ACM-SIAM Symp. on Discrete Algorithms},
  pages 13--21, January 2010.

\bibitem{BL07}
R.~Bar-Yehuda and J.~Laserson.
\newblock Exploiting locality: approximating sorting buffers.
\newblock {\em J. of Discrete Algorithms}, 5(4):729--738, 2007.

\bibitem{BB02}
D.~Blandford and G.~Blelloch.
\newblock Index compression through document reordering.
\newblock In {\em Data Compression Conference}, pages 342--351, 2002.

\bibitem{CMSS10}
H.-L. Chan, N. Megow, R. van Stee, and R. Sitters.
\newblock The sorting buffer problem is NP-hard.
\newblock {\tt CoRR, abs/1009.4355}, 2010.

\bibitem{ROADEF2005}
V-D.~Cung, A.~Nguyen, Y.~Khacheni, C.~Artigues, C.M.~Li, and B.~Penz.
\newblock Soci\'{e}t\'{e} fran\c{c}aise de Recherche Op\'{e}rationnelle et Aide \`{a} la D\'{e}cision
(ROADEF) Challenge 2005.
\newblock {\tt http://challenge.roadef.org/2005/en/}

\bibitem{ERW07}
M.~Englert, H.~R\"{a}cke, and M.~Westermann.
\newblock Reordering buffers for general metric spaces.
\newblock In {\em Proc. of the 39th Ann. ACM Symp. on Theory of Computing},
  pages 556--564, 2007.

\bibitem{EW05}
M.~Englert and M.~Westermann.
\newblock Reordering buffer management for non-uniform cost models.
\newblock In {\em Proc. of the 32nd Ann. Int'l Colloq. on
  Algorithms, Langauages, and Programming}, pages 627--638, 2005.

\bibitem{GS07}
I.~Gamzu and D.~Segev.
\newblock Improved online algorithms for the sorting buffer problem.
\newblock In {\em Proc. of the 24th Ann. Int'l Symp. on Theoretical
  Aspects of Computer Science}, pages 658--669, 2007.

\bibitem{GSV04}
K.~Gutenschwager, S.~Spiekermann, and S.~Vos.
\newblock A sequential ordering problem in automotive paint shops.
\newblock {\em Int'l J. of Production Research, 42(9):1865--1878},
  2004.

\bibitem{KP06}
R.~Khandekar and V.~Pandit.
\newblock Online sorting buffers on line.
\newblock In {\em Proc. of the 23rd Ann. Int'l Symp. on Theoretical
  Aspects of Computer Science}, pages 584--595, 2006.

\bibitem{KP04}
J.~Kohrt and K.~Pruhs.
\newblock Constant approximation algorithm for sorting buffers.
\newblock In {\em Proc. of the 6th Latin American Symp. on Theoretical
  Informatics}, pages 193--202, Buenos Aires, Argentina, 2004.

\bibitem{KRSW04}
J.~Krokowski, H.~R{\"a}cke, C.~Sohler, and M.~Westermann.
\newblock Reducing state changes with a pipeline buffer.
\newblock In {\em Proc. of the 9th Int'l Workshop on Vision, Modeling and Visualization},
  page 217, 2004.

\bibitem{LM09}
J.~Li and S.M.~Meerkov.
\newblock {\em Production Systems Engineering}.
\newblock Springer, 2009.

\bibitem{RSW02}
H.~R\"{a}cke, C.~Sohler, and M.~Westermann.
\newblock Online scheduling for sorting buffers.
\newblock In {\em Proc. of the 10th Ann. European Symp. on Algorithms},
  pages 820--832, 2002.

\bibitem{SGG09}
A.~Silberschatz, P.~Galvin, and G.~Gagne.
\newblock {\em Operating System Concepts}, 8th edition.
\newblock J. Wiley, 2009.

\end{thebibliography}
\end{document}